\documentclass[12pt,reqno]{article}

\usepackage[usenames]{color}
\usepackage{amssymb}
\usepackage{amsmath}
\usepackage{amsthm}
\usepackage{amsfonts}
\usepackage{amscd}
\usepackage{graphicx}
\usepackage{xcolor}

\usepackage{enumerate}

\usepackage[colorlinks=true,
linkcolor=webgreen,
filecolor=webbrown,
citecolor=webgreen]{hyperref}

\definecolor{webgreen}{rgb}{0,.5,0}
\definecolor{webbrown}{rgb}{.6,0,0}

\usepackage{color}
\usepackage{fullpage}
\usepackage{float}

\usepackage{graphics}
\usepackage{latexsym}
\usepackage{epsf}
\usepackage{breakurl}
\usepackage{fullpage}

\newcommand{\seqnum}[1]{\href{https://oeis.org/#1}{\rm \underline{#1}}}

\def\suchthat{ \, : \,}

\DeclareMathOperator{\per}{per}
\DeclareMathOperator{\cexp}{ce}

\newenvironment{smallarray}[1]
{\null\,\vcenter\bgroup\scriptsize
\arraycolsep=.13885em
\hbox\bgroup$\array{@{}#1@{}}}
{\endarray$\egroup\egroup\,\null}

\def\Zee{\mathbb{Z}}

\begin{document}

\theoremstyle{plain}
\newtheorem{theorem}{Theorem}
\newtheorem{corollary}[theorem]{Corollary}
\newtheorem{lemma}[theorem]{Lemma}
\newtheorem{proposition}[theorem]{Proposition}

\theoremstyle{definition}
\newtheorem{definition}[theorem]{Definition}
\newtheorem{example}[theorem]{Example}
\newtheorem{conjecture}[theorem]{Conjecture}
\newtheorem{problem}[theorem]{Problem}

\theoremstyle{remark}
\newtheorem{remark}[theorem]{Remark}

\title{Properties of a Ternary Infinite Word}

\author{
James Currie\footnote{
Department of Math/Stats, University of Winnipeg, 515 Portage Ave., Winnipeg, MB, R3B 2E9, Canada; e-mail
\href{mailto:j.currie@uwinnipeg.ca}{\tt j.currie@uwinnipeg.ca}.}
\and
Pascal Ochem\footnote{LIRMM, CNRS, Universit\'e de Montpellier, France; e-mail \href{ochem@lirmm.fr}{\tt ochem@lirmm.fr}.}
\and
Narad Rampersad\footnote{Department of Math/Stats, University of Winnipeg, 515 Portage Ave., Winnipeg, MB, R3B 2E9, Canada; e-mail
\href{mailto:n.rampersad@uwinnipeg.ca}{\tt n.rampersad@uwinnipeg.ca}.} \and
Jeffrey Shallit\footnote{
School of Computer Science,
University of Waterloo,
Waterloo, ON  N2L 3G1,
Canada; e-mail
\href{mailto:shallit@uwaterloo.ca}{\tt shallit@uwaterloo.ca}.}}

\maketitle

\begin{abstract}
We study the properties of the ternary infinite word
$${\bf p} = 012102101021012101021012 \cdots,$$ that is, the fixed point of the map $h: 0 \rightarrow 01$, $1 \rightarrow 21$, $2 \rightarrow 0$.   We determine its factor complexity, critical exponent, and prove that it is $2$-balanced.   We compute its abelian complexity and determine the lengths of its bispecial factors.  Finally, we give a characterization of $\bf p$ in terms of avoided factors.
\end{abstract}

\section{Introduction}

One of the themes of combinatorics on words is the study of particular
infinite words with interesting properties.   For example, in one of the very earliest results in this area, Thue proved that the
Thue-Morse word 
$$ {\bf t} = 0110100110010110 \cdots $$
avoids {\it overlaps}:  factors of the form $axaxa$ with $a$ a single letter
and $x$ a possibly empty word \cite{Thue:1912,Berstel:1995}.  He
also proved that the word
$$ {\bf vtm} = 2102012101202102012021 \cdots,$$
avoids {\it squares}:  factors of the form $yy$ with $y$ nonempty.

More generally, one can study other kinds of repetitions.
We
say that a finite word $w=w[1..n]$ has {\it period\/}
$p\geq 1$ if $w[i]=w[i+p]$ for $1 \leq i \leq n-p$.
The smallest period of a word $w$ is called
{\it the\/} period, and we write it as
$\per(w)$.  The {\it exponent\/} of a finite word
$w$, written $\exp(w)$ is defined to be
$|w|/\per(w)$.   We say a word
(finite or infinite) is
{\it $\alpha$-free\/} if the exponent of all its nonempty factors is $>\alpha$.   We say a word is {\it $\alpha^+$-free\/} if the exponent of all its nonempty factors is $\geq\alpha$.  The {\it critical exponent\/} of
a finite or infinite word $x$ is the supremum, over all
nonempty finite factors $w$ of $x$, of
$\exp(w)$; it is written $\cexp(x)$. 
The critical exponent of a word can be either rational or irrational.   If it is rational, then it can either be attained by a particular finite factor, or not attained.  For example, the critical exponent of both $\bf t$ and $\bf vtm$ is $2$, but it is attained in the former case and not attained in the latter.  If the critical exponent $\alpha$ is attained, we typically write it as $\alpha^+$.  An overlap is a $2^+$ power, so the Thue-Morse word is $2^+$-free.

The Fibonacci word
$$ {\bf f} = 010010100100101001010 \cdots $$
is the fixed point of the morphism $0 \rightarrow 01$, $1 \rightarrow 0$.
Karhum\"aki proved \cite{Karhumaki:1983} that $\bf f$ has no fourth powers (i.e., blocks of the
form $xxxx$, with $x$ nonempty), and  Mignosi and Pirillo \cite{Mignosi&Pirillo:1992} proved
that the critical exponent of $\bf f$ is $(5+\sqrt{5})/2$.

Another aspect of infinite words that has been studied is 
{\it balance}.  We say that a finite or infinite word $x$ is
$t$-balanced if for all equal-length factors $y, z$ of $x$,
and all letters $a$,
the inequality $\left| |y|_a - |z|_a \right| \leq t$ is satisfied.   
As is well-known, the Fibonacci word $\bf f$ (and more generally, every Sturmian word) is $1$-balanced
\cite{Morse&Hedlund:1940,Coven&Hedlund:1973}.

A third aspect is {\it factor complexity}, also called {\it subword complexity}. For an infinite word
$\bf x$, the factor complexity function $\rho_{\bf x} (n)$ counts
the number of distinct factors of length $n$ in $\bf x$.   
Morse and Hedlund \cite{Morse&Hedlund:1940} proved that $\rho_{\bf f} (n) = n+1$ for all
$n \geq 0$.   There is also the abelian analogue of factor
complexity, where we count two factors as the same if they are
permutations of each other \cite{Richomme&Saari&Zamboni:2011}.

In this note we study various aspects of the word
$$ {\bf p} = 012102101021012101021012 \cdots ,$$
fixed point of the map $h$ sending $0 \rightarrow 01$, $1 \rightarrow 21$, $2 \rightarrow 0$.   This word is not automatic (because, as we will see,
letters occur with irrational densities).   It is not
Sturmian (because it is over a $3$-letter alphabet).  Neither is it episturmian \cite{Glen&Justin:2009},
because its set of subwords is not closed under reversal: $\bf p$ contains $02$,
but avoids $20$.

We determine its factor complexity, its critical exponent, and prove
that it is $2$-balanced.   A novel aspect of our work is that much of it is carried out using the
{\tt Walnut} theorem-prover \cite{Mousavi:2016,Shallit:2022}.   This software tool can prove or disprove assertions phrased in first-order logic about automatic sequences and their generalizations.   These ideas were used previously to study the Tribonacci word \cite{Mousavi&Shallit:2015}, but the word $\bf p$ provides some new complications.   All the files required to carry out the computations are available at the last author's website:\\
\centerline{\url{ https://cs.uwaterloo.ca/~shallit/papers.html} \ .} \\

The word $\bf p$ has been studied previously.  For example, it is sequence
\seqnum{A287072} in the OEIS.    It also appears implicitly in \cite{Cassaigne&Labbe&Leroy:2017},
where it is (up to renaming of the letters)
the fixed point of the word $c_2 c_1$.

The results in this paper are applied to an
avoidability problem in the companion paper
\cite{Currie&Mol&Ochem&Rampersad&Shallit:2022}.

\section{A Pisot numeration system}\label{pisot}

The properties of $\bf p$ are intimately related to a particular numeration system $P4$, which we discuss now.

Consider the following linear recurrence:
$$X_1 = 1, X_2 = 2, X_3 = 4, X_4 = 7, \text{ and }
	X_n = X_{n-1} + X_{n-2} + X_{n-4} \text{ for $n \geq 0$ }.$$
Here are the first few terms of this recurrence:
\begin{table}[H]
\begin{center}
\begin{tabular}{c|cccccccccccccc}
$n$ & 1 & 2 & 3 & 4 & 5 & 6 & 7 & 8 & 9 & 10 & 11 & 12 & 13 & 14 \\
\hline
$X_n$ & 1 & 2 & 4 & 7 & 12 & 21 & 37 & 65 & 114 & 200 & 351 & 616 & 1081 & 1897 
\end{tabular}
\end{center}
\caption{The recurrence $X_n$}
\end{table}
This is sequence \seqnum{A005251} in the OEIS.  It is easy to verify it also satisfies the simpler recurrence
$X_n = 2X_{n-1} - X_{n-2} + X_{n-3}$ for $n \geq 4$.

We can consider representing natural numbers as a sum of the $X_i$,
as follows: $N = \sum_{1 \leq i \leq t} e_i X_i$, where 
$e_i \in \{0,1\}$.   If we impose the following two rules
on such a representation, namely
\begin{itemize}
\item[(a)]  $(e_i, e_{i+2}, e_{i+3}) \not= (1,0,1,1)$; and
\item[(b)]  $(e_i, e_{i+1}, e_{i+2}) \not= (1,1,1)$.
\end{itemize}
then this representation is unique, and can be written as
a binary string $(N)_{P} := e_t e_{t-1} \cdots e_2 e_1$.  For example, here are the
first few representations of numbers in this numeration system:
\begin{table}[H]
\begin{center}
\begin{tabular}{c|c||c|c}
$n$ & $(n)_P$ & $n$ & $(n)_P$ \\
\hline
 1 &  1      & 13 &  10001  \\
 2 &  10     & 14 &  10010  \\
 3 &  11     & 15 &  10011  \\
 4 &  100    & 16 &  10100  \\
 5 &  101    & 17 &  10101  \\
 6 &  110    & 18 &  10110  \\
 7 &  1000   & 19 &  11000  \\
 8 &  1001   & 20 &  11001  \\
 9 &  1010   & 21 &  100000 \\
10 &  1011   & 22 &  100001 \\
11 &  1100   & 23 &  100010 \\
12 &  10000  & 24 &  100011 \\
\end{tabular}
\end{center}
\end{table}
These representations are, in fact, the ones resulting
by applying the greedy algorithm,
and the conditions (a) and (b) follow from a theorem of
Fraenkel \cite{Fraenkel:1985}.

A representation in this system $P4$ is valid if and only if it
contains no occurrence of $111$ or $1101$.  It follows that
the language of all valid representations is recognized
by the automaton in Figure~\ref{valid}.
\begin{figure}[H]
\begin{center}
\includegraphics[width=5in]{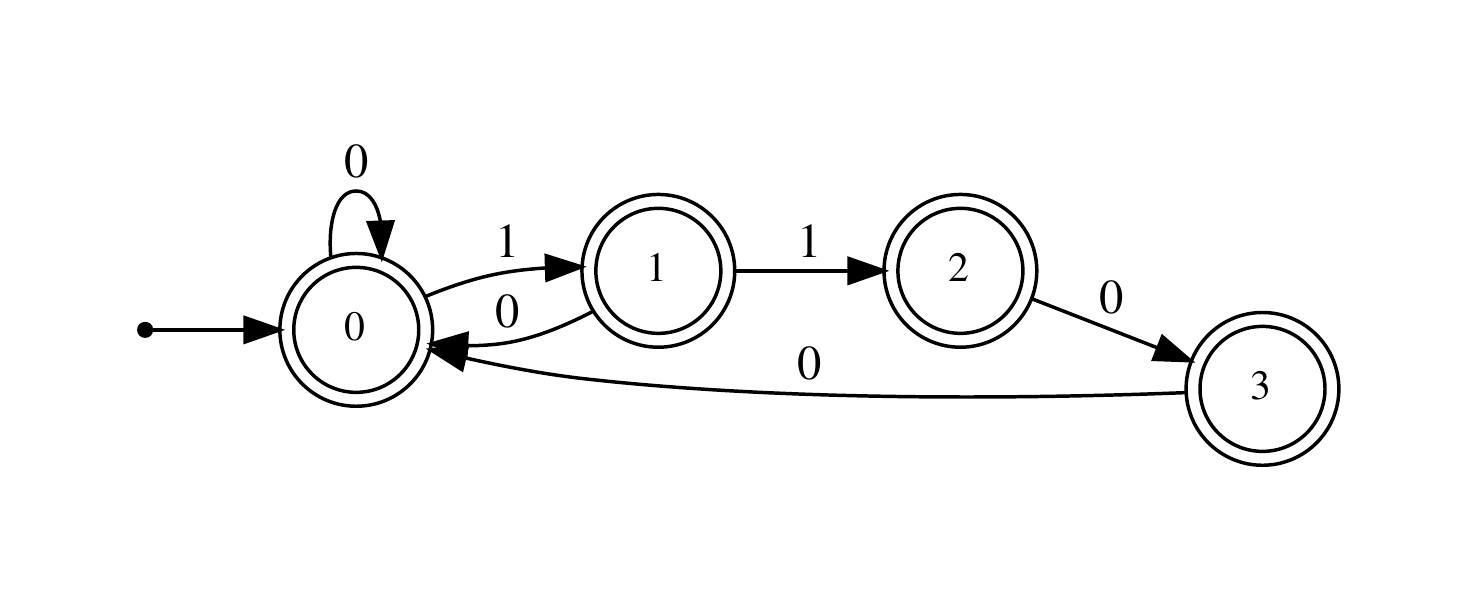}
\end{center}
\caption{Automaton recognizing the valid representations.}
\label{valid}
\end{figure}

The characteristic polynomial of the recurrence is
$X^4 - X^3 -X^2 -1 = (X+1) (X^3-2X^2+X-1)$.   The
second term has one real zero, namely
$$ \beta_1 = 
{{(100 + 12 \sqrt{69})^{1/3}} \over 6}
+ 
{2 \over {3 (100+12 \sqrt{69})^{1/3}}} + 2/3
\doteq 1.7548776662466927600495,$$
and two imaginary zeros that lie inside the
unit circle.  Therefore $\beta_1$ is a Pisot number,
and so, by the results
in \cite{Frougny&Solomyak:1996,Bruyere&Hansel:1997}
we know that there is a finite automaton $A$ recognizing
the addition relation $x+y=z$, where $x,y,z$ are
represented in the numeration system described above,
with inputs in parallel and the shorter padded with
leading zeros.  Furthermore, there is an algorithm
to compute $A$ for any Pisot number.

Instead of applying this algorithm to compute $A$,
we took a different approach.  Namely, we ``guessed''
the adder using the Myhill-Nerode theorem, and then verified
it using an incrementer constructed and verified by hand.
This incrementer computes the relation $y=x+1$, and is
illustrated in Figure~\ref{inc-fig}.
Here the inputs are $x$ and $y$
in parallel, both represented in the numeration system $P4$.
\begin{figure}[H]
\begin{center}
\includegraphics[width=6in]{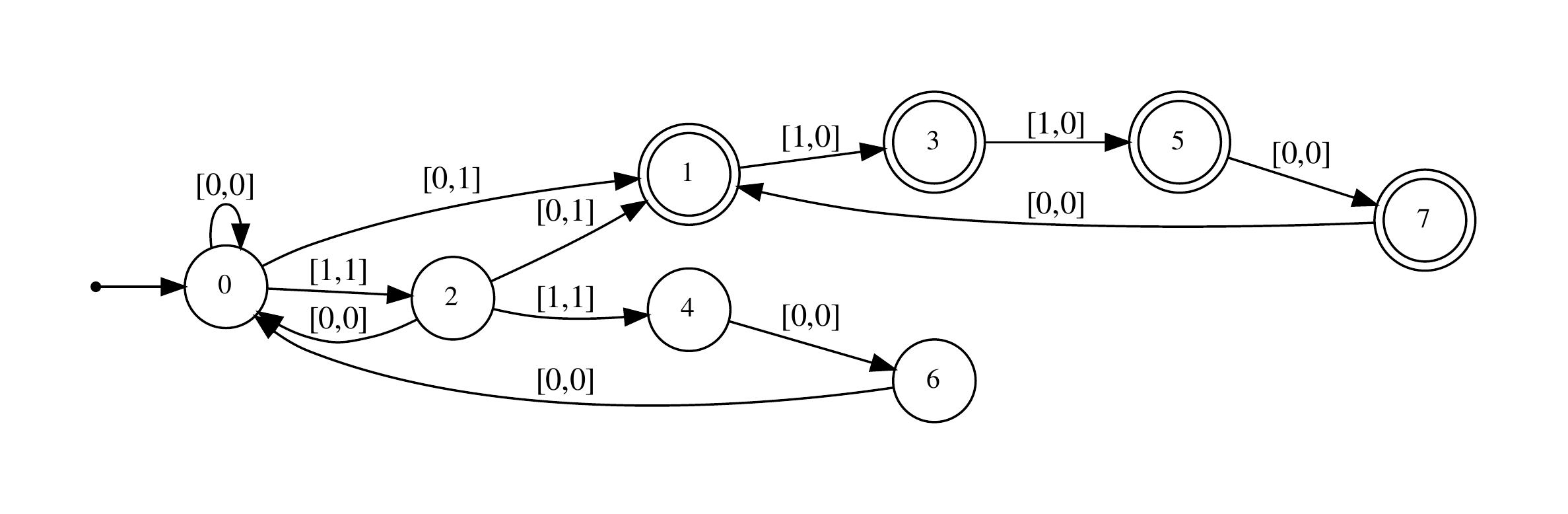}
\end{center}
\caption{Incrementer in the numeration system.}
\label{inc-fig}
\end{figure}
Once we have the incrementer, the correctness of the
adder can be verified as done in 
\cite{Mousavi&Schaeffer&Shallit:2016}.  The adder
has 64 states.

It turns out that the word $\bf p$ is automatic
in this numeration system; this means it is computed by a deterministic finite automaton with output (DFAO) taking the $P4$ representation of $n$ as input
and outputting (in the last state reached) the value of ${\bf p}[n]$.
The DFAO computing it is
depicted in Figure~\ref{paut}.
\begin{figure}[H]
\begin{center}
\includegraphics[width=6in]{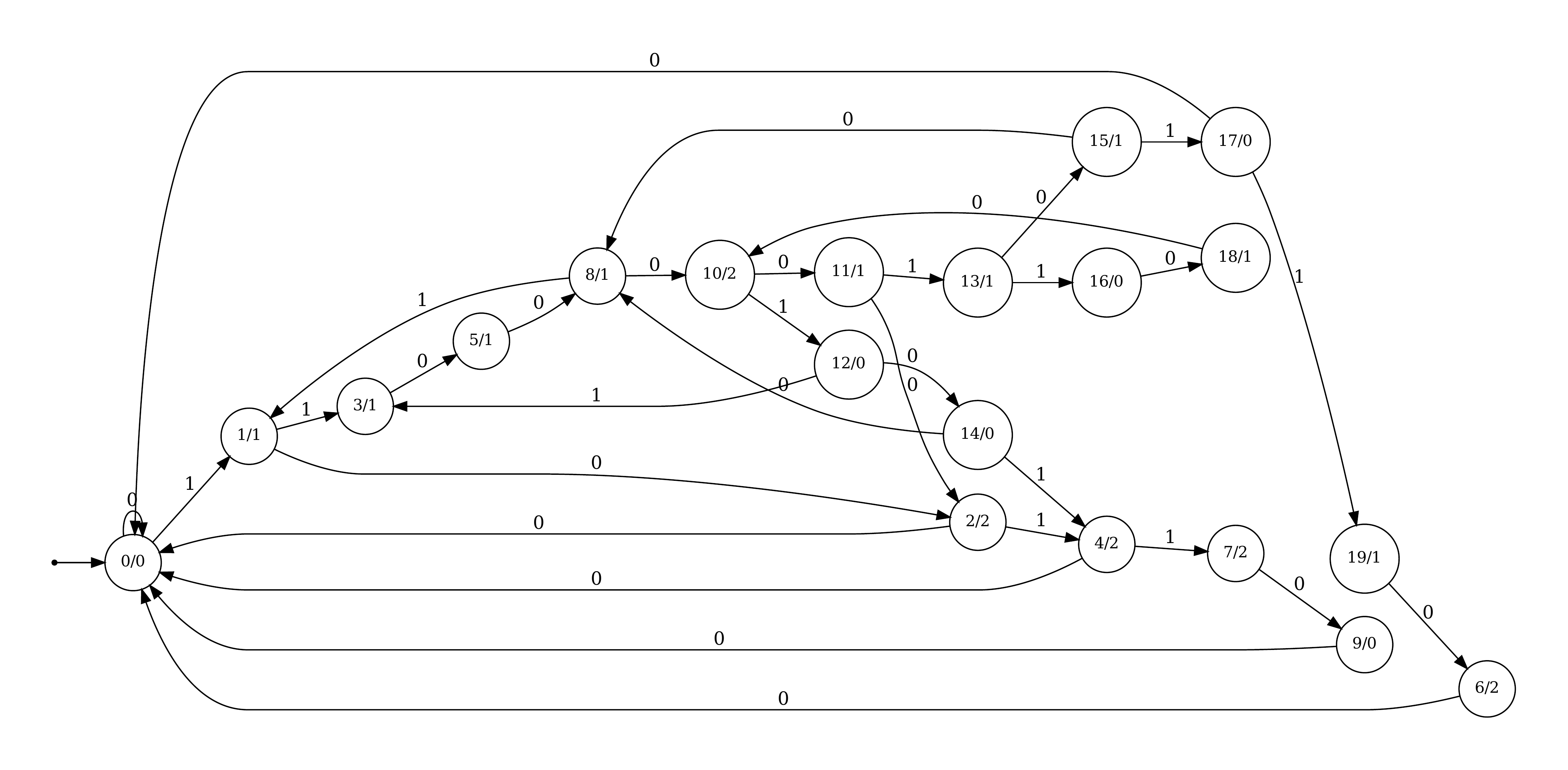}
\end{center}
\caption{DFAO computing $\bf p$ in the numeration system.}
\label{paut}
\end{figure}

This can be verified as follows.
Let $h:0 \rightarrow 01$, $1 \rightarrow 21$, $2 \rightarrow 0$,
and let the automaton be represented by a morphism $\varphi$ encoding the transitions of the automaton, and a coding
$\tau$ giving the outputs for each state.   Then we need to verify the following identities:
\begin{align*}
\tau(\varphi^n(0)) &= h^n(0) &
\tau(\varphi^n(1)) &= h^n(1) \\
\tau(\varphi^n(23)) &= h^n(21) &
\tau(\varphi^n(45)) &= h^n(21) \\
\tau(\varphi^n(78)) &= h^n(21) &
\tau(\varphi^n(9A)) &= h^n(02) \\
\tau(\varphi^n(BC)) &= h^n(10) &
\tau(\varphi^n(2DE3)) &= h^n(2101) \\
\tau(\varphi^n(4FG8)) &= h^n(2101) &
\tau(\varphi^n(HIA)) &= h^n(012) \\
\tau(\varphi^n(JA)) &= h^n(12) &
\tau(\varphi^n(6)) &= h^n(2) \\
\end{align*}
where $A = 10$, $B = 11$, etc.  This can be done by a
tedious induction on $n$, which we omit.   Just to demonstrate one needed identity:
\begin{align*}
\tau(\varphi^n(9A)) &= \tau(\varphi^{n-1}(0BC)) \\
&= \tau(\varphi^{n-1} (0)) \tau(\varphi^{n-1} (BC)) \\
&= h^{n-1} (0) h^{n-1} (10) \\
&= h^{n-1} (010) \\
&= h^n(02),
\end{align*}
as desired.

\section{Factor complexity of $\bf p$}

In this section we prove that the factor complexity of $\bf p$ is $2n+1$.   This is also a consequence of more general results of \cite{Cassaigne&Labbe&Leroy:2017}.   Also see
\cite{Cassaigne&Labbe&Leroy:2022}.

There is a well-established computational method for determining factor complexity, as discussed in \cite{Shallit:2021f}.
The first step is to create an automaton that, given integers $i,j,n$ as input, decides if
${\bf p}[i..i+n-1]={\bf p}[j..j+n-1]$.   Normally we would do this with the following {\tt Walnut} command, where
{\tt PI} is a file containing the automaton in Figure~\ref{paut}:
\begin{verbatim}
def pisotfaceq "?msd_pisot4 At (t<n) => PI[i+t]=PI[j+t]":
\end{verbatim}
However, in this case, the attempt fails.  {\tt Walnut}
tries to determinize an automaton with 37351 states, and fails
even with 5 Terabytes of storage and many hours of computation, so we need a different approach.

Instead, we use the approach discussed in \cite[\S 6.3]{Shallit:2022}.  We
``guess'' an automaton for ``pisotfaceq" using the Myhill-Nerode theorem.
States are labeled with prefixes of Pisot representations.  Two states
are guessed to be the same if all suffixes of length $\leq 5$ (representing $1 + 8^2 + \cdots + 8^5 = 37449$
words) give the same result.
This gives us an automaton with 1080 states that we can represent
as a file called {\tt pisi.txt}.

Next, we can use induction and
{\tt Walnut} together to prove that our guess is correct.  We
do this as follows:
\begin{verbatim}
eval zeros "?msd_pisot4 Ai,j $pisi(i,j,0)":
eval induc "?msd_pisot4 Ai,j,n ($pisi(i,j,n) & PI[i+n]=PI[j+n]) 
   => $pisi(i,j,n+1)":
\end{verbatim}
Both of these return {\tt TRUE},
so our guess is correct.

Next, we create a linear representation $(v, \zeta, w)$ for the subword complexity function, using the
following {\tt Walnut} command:
\begin{verbatim}
eval pisotsc n "?msd_pisot4 Aj j<i => ~$pisi(i,j,n)":
\end{verbatim}
This gives us a linear representation of rank 131.  When we minimize it using
a Maple program, we get
this linear representation of rank 16:
\begin{align*}
v &= \left[
\begin{smallarray}{cccccccccccccccc}
1&0&0&0&0&0&0&0&0&0&0&0&0&0&0&0
\end{smallarray} \right] 
\end{align*}
\begin{align*}
\zeta(0) &= \left[
\begin{smallarray}{cccccccccccccccc}
1&  0&  0&  0&  0&  0&  0&  0&  0&  0&  0&  0&  0&  0&  0&  0\\
   0&  0&  1&  0&  0&  0&  0&  0&  0&  0&  0&  0&  0&  0&  0&  0\\
   0&  0&  0&  0&  1&  0&  0&  0&  0&  0&  0&  0&  0&  0&  0&  0\\
   0&  0&  0&  0&  0&  0&  1&  0&  0&  0&  0&  0&  0&  0&  0&  0\\
   0&  0&  0&  0&  0&  0&  0&  1&  0&  0&  0&  0&  0&  0&  0&  0\\
  -1&  0&  1&  0&  0&  0&  0&  1&  0&  0&  0&  0&  0&  0&  0&  0\\
  -1&  0&  0&  0&  1&  0&  0&  1&  0&  0&  0&  0&  0&  0&  0&  0\\
  -1&  0&  1&  0& -1&  0&  0&  2&  0&  0&  0&  0&  0&  0&  0&  0\\
  -2&  0&  2&  0& -1&  0&  0&  2&  0&  0&  0&  0&  0&  0&  0&  0\\
   0&  0&  0&  0&  0&  0&  0&  0&  0&  0&  0&  0&  1&  0&  0&  0\\
  -4&  0&  3&  0& -1&  0&  0&  3&  0&  0&  0&  0&  0&  0&  0&  0\\
   0&  0&  0&  0&  0&  0&  0&  0&  0&  0&  0&  0&  0&  0&  1&  0\\
  -5&  0&  2&  0&  0&  0&  0&  4&  0&  0&  0&  0&  0&  0&  0&  0\\
   0&  0&  0&  0&  0&  0&  0&  0&  0&  0&  0&  0&  0&  0&  0&  1\\
  -8&  0&  3&  0&  0&  0&  0&  6&  0&  0&  0&  0&  0&  0&  0&  0\\
 -13&  0&  5&  0& -1&  0&  0& 10&  0&  0&  0&  0&  0&  0&  0&  0
\end{smallarray} \right] 
& \zeta(1) &= \left[
\begin{smallarray}{cccccccccccccccc}
0&1&0&0&0&0&0&0&0&0&0&0&0&0&0&0\\
 0&0&0&1&0&0&0&0&0&0&0&0&0&0&0&0\\
 0&0&0&0&0&1&0&0&0&0&0&0&0&0&0&0\\
 0&0&0&0&0&0&0&0&0&0&0&0&0&0&0&0\\
 0&0&0&0&0&0&0&0&1&0&0&0&0&0&0&0\\
 0&0&0&0&0&0&0&0&0&1&0&0&0&0&0&0\\
 0&0&0&0&0&0&0&0&0&0&0&0&0&0&0&0\\
 0&0&0&0&0&0&0&0&0&0&1&0&0&0&0&0\\
 0&0&0&0&0&0&0&0&0&0&0&1&0&0&0&0\\
 0&0&0&0&0&0&0&0&0&0&0&0&0&0&0&0\\
 0&0&0&0&0&0&0&0&0&0&0&0&0&1&0&0\\
 0&0&0&0&0&0&0&0&0&0&0&0&0&0&0&0\\
 0&0&0&0&0&0&0&0&0&0&0&0&0&0&0&0\\
 0&0&0&0&0&0&0&0&0&0&0&0&0&0&0&0\\
 0&0&0&0&0&0&0&0&0&0&0&0&0&0&0&0\\
 0&0&0&0&0&0&0&0&0&0&0&0&0&0&0&0
\end{smallarray} \right] 
& w &= \left[
\begin{smallarray}{c}
1\\
  3\\
  5\\
  7\\
  9\\
 11\\
 13\\
 15\\
 17\\
 21\\
 27\\
 31\\
 37\\
 49\\
 55\\
 87
\end{smallarray} \right] \\
\end{align*}

With all these tools at our 
disposal, we can now prove the
desired result.
\begin{theorem}
The factor complexity of 
$\bf p$ is $2n+1$.
\label{subword}
\end{theorem}

\begin{proof}
It remains to show that the
linear representation
$(v, \zeta, w)$ computes the
function $2n+1$.  To do this, 
we compute a linear representation
for $2n+1$ using
the following {\tt Walnut}
command:
\begin{verbatim}
eval c2n1 n "?msd_pisot4 i<2*n+1":
\end{verbatim}
This gives us a linear representation
$(v', \zeta', w')$ of rank 108.
When we minimize it, we get {\it exactly\/} the same linear 
representation 
$(v, \zeta, w)$.  This shows they are the same and proves the result.
\end{proof}

\section{Critical exponent of $\bf p$}

In this section we determine
the critical exponent of $\bf p$.

\begin{theorem}\label{gamma}
The critical exponent of
$\bf p$ is $\gamma + 1 = (8-\beta_1+2\beta_1^2)/5
\doteq 2.480862716147236962394265321$.    Here $\gamma+1$ is
the real zero of
$5X^3-26X^2+43X-23$.
\end{theorem}

\begin{proof}
We use the strategy previously
employed for the Tribonacci word
\cite{Mousavi&Shallit:2015}.

Now that we have a DFAO for $\bf p$, we can find all the
periods corresponding to overlaps.  We do this with the
following Walnut command.
\begin{verbatim}
def pisotlargepow "?msd_pisot4 Ei (n>=1) & At (t>=i & t<=i+n) => PI[t]=PI[t+n]":
\end{verbatim}
This gives the automaton depicted in Figure~\ref{psq}.
\begin{figure}[H]
\begin{center}
\includegraphics[width=5in]{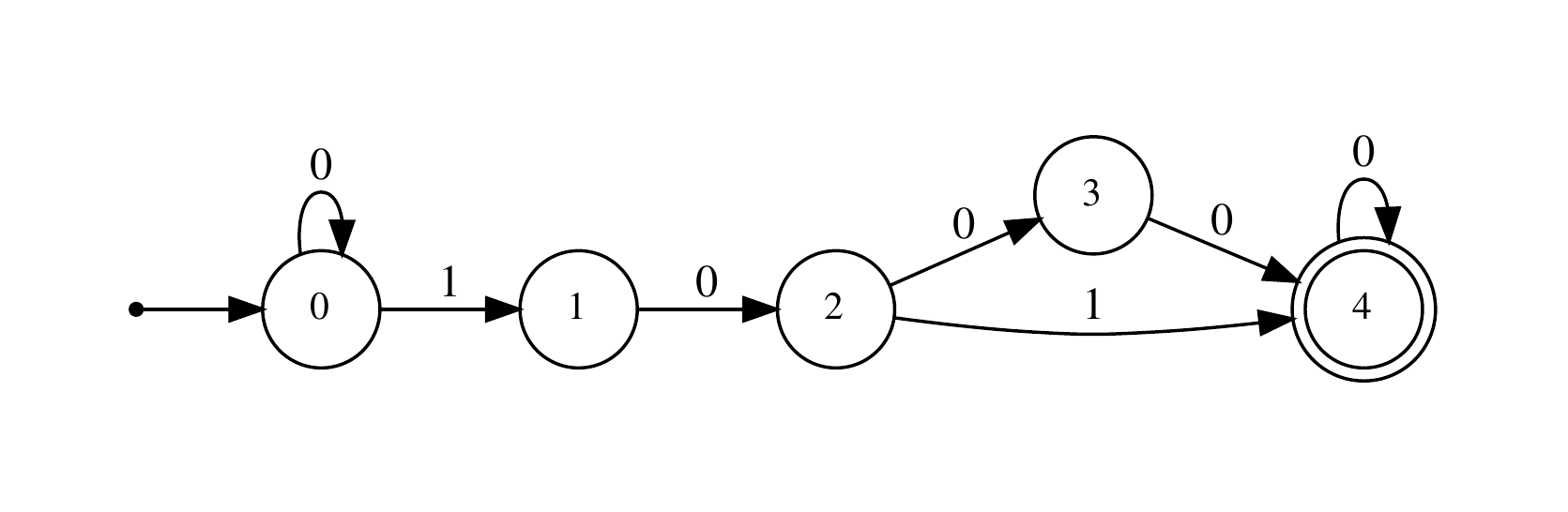}
\end{center}
\caption{DFAO computing orders of overlaps in $\bf p$.}
\label{psq}
\end{figure}
Constructing this automaton involved determinizing an NFA with 10,859 states and
minimizing a DFA with 13,114,119 states; the calculation needed 90 Gigs of storage and 1100 seconds of CPU time.  

As you can see from inspecting it, the periods are given by 
the representations $1010^*$ and $10000^*$.
We now compute the maximal repetitions in $\bf p$ by looking only at
the periods given above:
\begin{verbatim}
def maximalreps "?msd_pisot4 Ei
(PI[i+n] != PI[i+n+p]) & $pisotlargepow(p) &
    (Aj (j<n & $pisotlargepow(p)) => PI[i+j] = PI[i+j+p])":
def highestpow "?msd_pisot4 (p>=1) & $pisotlargepow(p) &
    $maximalreps(n,p) & (Am $maximalreps(m,p) => m <= n)":
\end{verbatim}
This gives the automaton in Figure~\ref{highest}.  It recognizes all
pairs $(n,p)$ such that $n/p + 1$ is a maximal power in $\bf p$ with
period $p$, for the $p$ given above.
\begin{figure}[H]
\begin{center}
\includegraphics[width=6in]{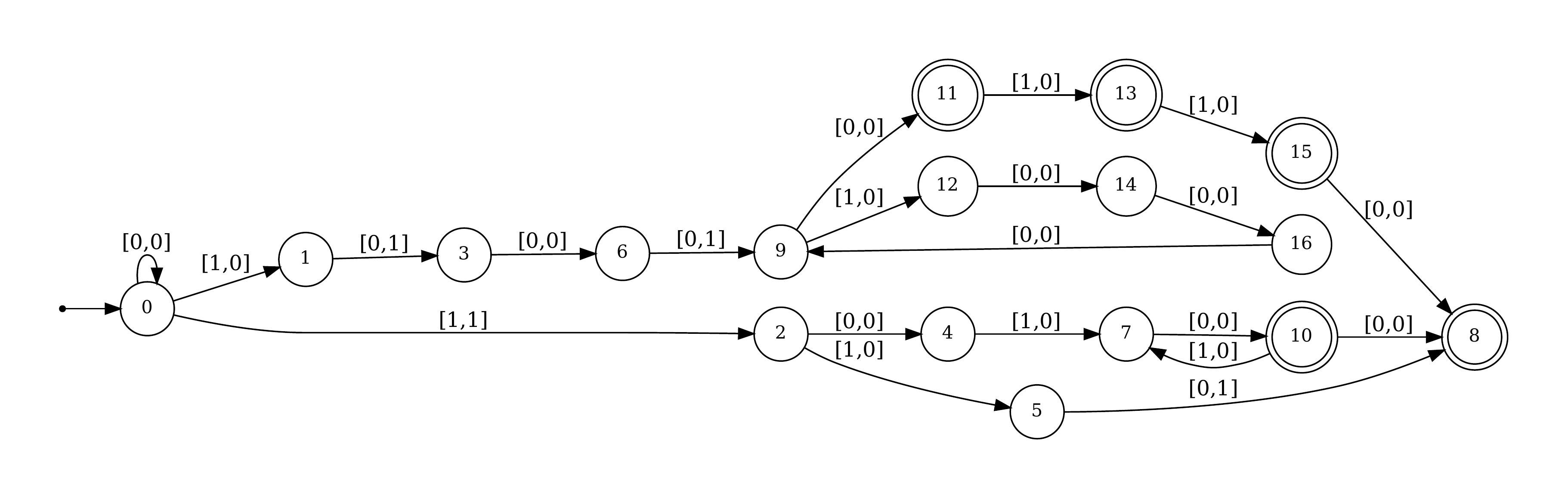}
\end{center}
\caption{DFAO computing maximal powers}
\label{highest}
\end{figure}
By inspection we see these are as follows:
\begin{align*}
& [1,1][1,0][0,1] \\
& [1,1][0,0][1,0] ([0,0][1,0])^* [0,0] \\
& [1,1][0,0][1,0] ([0,0][1,0])^* [0,0][0,0] \\
& [1,0][0,1][0,0][0,1] ([1,0][0,0][0,0][0,0])^* [0,0] \\
& [1,0][0,1][0,0][0,1] ([1,0][0,0][0,0][0,0])^* [0,0] [1,0] \\
& [1,0][0,1][0,0][0,1] ([1,0][0,0][0,0][0,0])^* [0,0] [1,0] [1,0] \\
& [1,0][0,1][0,0][0,1] ([1,0][0,0][0,0][0,0])^* [0,0] [1,0] [1,0] [0,0]
\end{align*}
These correspond to, respectively, exponents of 
\begin{align*}
6/5  \\
(\sum_{1 \leq i \leq n+2} X_{2i})/X_{2n+4}, \quad n \geq 0 \\
(\sum_{1 \leq i \leq n+2} X_{2i+1})/X_{2n+5}, \quad n \geq 0 \\
(\sum_{1 \leq i \leq n+1} X_{4i+1})/(X_{4n+4}+X_{4n+2}), \quad n \geq 0 \\
(1+\sum_{1 \leq i \leq n+1} X_{4i+2})/(X_{4n+5}+X_{4n+3}), \quad n\geq 0 \\
(3+\sum_{1 \leq i \leq n+1} X^{4i+3})/(X_{4n+6}+X_{4n+4}), \quad n \geq 0 \\
(6+\sum_{1 \leq i \leq n+1} X^{4i+4})/(X_{4n+7}+X_{4n+5}), \quad n \geq 0.
\end{align*}

It now remains to check that these expressions are all less than
$\gamma$ and the last six tend to $\gamma$ from below.
We explain how to do this for the second expression; the others
are similar.

First, by the standard theory of linear recurrences
\cite{Everest&vanderPoorten&Shparlinski&Ward:2003}, we know that
$$X_n = \alpha_1 \beta_1^n + \alpha_2 \beta_2^n + \alpha_3 \beta_3^n,$$
where $\beta_1, \beta_2, \beta_3$ are the zeros of
$X^3-2X^2+X-1$, and $\beta_1 \doteq 1.7548776662466927600495$ is the unique real zero.
By solving the appropriate linear system, we
find that $\alpha_1 = (\beta_1^2 + 6\beta_1 + 3)/23$.
Furthermore,
$|\alpha_2| + |\alpha_3| < 1/2$.
Since $|\beta_2| = |\beta_3|$ and the
product of the zeros is equal to the constant term of the defining polynomial, which is $1$,
we get
    $|\beta_2|^2 \beta_1 = 1$,
which gives
    $|\beta_2| = \sqrt{1/\beta_1} = \beta_1 - 1$
by the defining equation. 
 It follows that
$$ |X_n - \alpha_1 \beta_1^n| = |\alpha_2 \beta_2^n + \alpha_3 \beta_3^n|
\leq (|\alpha_2|+|\alpha_3|)(\beta_1-1)^n < (\beta_1-1)^n.$$

Next, one can prove by induction (or using Walnut!) that
$$\sum_{1 \leq i \leq n} X_{2i} = {{3X_{2n} - X_{2n+1} + 2X_{2n+2} - 6} \over 5}.$$
It follows that
\begin{align*}
\sum_{1 \leq i \leq n} X_{2i} & \leq
{{3X_{2n} - X_{2n+1} + 2X_{2n+2} - 6} \over 5} \\
&\leq {{3-\beta_1 + 2 \beta_1^2} \over 5} (\alpha_1 \beta_1^{2n}) + 6 (\beta_1-1)^{2n} - 6/5.
\end{align*}
Dividing by $X_{2n}$, we see that the quotient
tends to $\gamma$ from below, where
$\gamma = {{3-\beta_1 + 2\beta_1^2} \over 5}$.  Then $\gamma+1$ is the desired
critical exponent.
\end{proof}

\section{Synchronization and balance} 

We now show that the functions $c_i : n \rightarrow |{\bf p}[0..n-1]|_i$
are synchronized in this numeration system, for $i \in \{ 0,1, 2\}$.
This means that there exists an automaton $A_i$ taking $n$ and $x$ as inputs (in $P4$ representation) and accepting if $x = c_i (n)$,
for $i = 0,1,2$.  For more information about synchronization, see \cite{Shallit:2021h}.

To do so, we ``guess'' the automata for the 
$c_i$ using the Myhill-Nerode theorem
and then verify our guesses using {\tt Walnut}.
Here is the code for verification, where
{\tt psynch0, psynch1, psynch2} are the
guessed automata:
\begin{verbatim}
eval tmp0 "?msd_pisot4 An,x (($psynch0(n,x) & PI[n]=@0) => $psynch0(n+1,x+1)) 
& (($psynch0(n,x) & PI[n]!=@0) => $psynch0(n+1,x))":
eval tmp1 "?msd_pisot4 An,x (($psynch1(n,x) & PI[n]=@1) => $psynch1(n+1,x+1)) 
& (($psynch1(n,x) & PI[n]!=@1) => $psynch1(n+1,x))":
eval tmp2 "?msd_pisot4 An,x (($psynch2(n,x) & PI[n]=@2) => $psynch2(n+1,x+1)) 
& (($psynch2(n,x) & PI[n]!=@2) => $psynch2(n+1,x))":
\end{verbatim}
and all of these return
{\tt TRUE}.

Next, we create synchronized
automata for
$|{\bf p}[i..i+n-1]|_i$:
\begin{verbatim}
def pcount0 "?msd_pisot4 Ex,y $psynch0(i,x) & $psynch0(i+n,y) & y=x+z":
# 1687 states
def pcount1 "?msd_pisot4 Ex,y $psynch1(i,x) & $psynch1(i+n,y) & y=x+z":
# 2626 states
def pcount2 "?msd_pisot4 Ex,y $psynch2(i,x) & $psynch2(i+n,y) & y=x+z":
# 2773 states
\end{verbatim}

We can now prove the following
result.
\begin{theorem}
The word $\bf p$ is
$2$-balanced.
\end{theorem}

\begin{proof}
We use the following 
{\tt Walnut} commands.
\begin{verbatim}
def twobalanced0 "?msd_pisot4 Ai,j,n,x,y ($pcount0(i,n,x) & 
   $pcount0(j,n,y)) => (y<=x+2 & x<=y+2)":
def twobalanced1 "?msd_pisot4 Ai,j,n,x,y ($pcount1(i,n,x) & 
   $pcount1(j,n,y)) => (y<=x+2 & x<=y+2)":
def twobalanced2 "?msd_pisot4 Ai,j,n,x,y ($pcount2(i,n,x) & 
   $pcount2(j,n,y)) => (y<=x+2 & x<=y+2)":
\end{verbatim}
{\tt Walnut} returns
{\tt TRUE} for all of these.
\end{proof}

\section{Abelian complexity}

We can compute the abelian complexity of $\bf p$
with {\tt Walnut} in much the same way that it
was done for the Tribonacci word in \cite{Mousavi&Shallit:2015}, with some minor
modifications.   For each $n$ and $i \in \{0,1,2\}$,
we compute
the vector $u_i (n) = \min_x |x|_i$, where the minimum is over all the length-$n$
factors of $\bf p$.  
\begin{verbatim}
def min0 "?msd_pisot4 Ei $pcount0(i,n,x) & Aj,y $pcount0(j,n,y) => y>=x":
# 169 states

def min1 "?msd_pisot4 Ei $pcount1(i,n,x) & Aj,y $pcount1(j,n,y) => y>=x":
# 169 states

def min2 "?msd_pisot4 Ei $pcount2(i,n,x) & Aj,y $pcount2(j,n,y) => y>=x":
# 223 states
\end{verbatim}

Once we have this, we can show
that 
\begin{multline*}
\psi({\bf p}[i..i+n-1])-(u_0(n), u_1(n), u_2(n)) \in
\{ (0,0,1), (0,0,2), (0,1,0), (0,1,1), (0,1,2),\\
(0,2,0), (0,2,1), (1,0,0), (1,0,1), (1,0,2), (1,1,0), (1,1,1), (1,2,0), (2,0,0), (2,0,1), (2,1,0) \}
\end{multline*}
for $i \geq 0$ and $n \geq 1$, as follows:
\begin{verbatim}
def validtriples "?msd_pisot4 Ei,n,a,b,c $pcount0(i,n,a+x) & $min0(n,a) &    
   $pcount1(i,n,b+y) & $min1(n,b) & $pcount2(i,n,c+z) & $min2(n,c)":
\end{verbatim}

Next, we can show that
$$ \{ \psi({\bf p}[i..i+n-1])-(u_0(n), u_1(n), u_2(n))
\suchthat i \geq 0 \} $$
is one of the following 18 possible sets:
\begin{align*}
S_1 &= \{(0, 0, 1),  (0, 1, 0),  (1, 0, 0) \} \\    
S_2 &= \{ (0, 1, 1),  (1, 0, 1),  (1, 1, 0) \}\\                                             S_3 &= \{ (0, 1, 1),  (1, 0, 1),  (1, 1, 0),  (2, 0, 0) \} \\
S_4 &=
  \{(0, 0, 2),  (0, 1, 1),  (1, 0, 1),  (1, 1, 0) \}  \\                                S_5 &= \{ (0, 1, 1),  (0, 2, 0),  (1, 0, 1),  (1, 1, 0) \} \\                                  S_6 &= \{ (0, 0, 2),  (0, 1, 1),  (1, 0, 1),  (1, 1, 0),  (2, 0, 0) \} \\
S_7 &= \{
  (0, 1, 2),  (1, 0, 2),  (1, 1, 1),  (2, 0, 1),  (2, 1, 0)    \} \\                    S_8 &= \{ (0, 2, 1),  (1, 1, 1),  (1, 2, 0),  (2, 0, 1),  (2, 1, 0)  \} \\                       S_{9} &= \{  (0, 1, 1),  (0, 2, 0),  (1, 0, 1),  (1, 1, 0),  (2, 0, 0) \} \\                       
S_{10} &= \{ (0, 0, 2),  (0, 1, 1),  (0, 2, 0),  (1, 0, 1),  (1, 1, 0) \} \\                       S_{11} &= \{ (0, 1, 2),  (0, 2, 1),  (1, 0, 2),  (1, 1, 1),  (1, 2, 0) \} \\                       
S_{12} &= \{ (0, 1, 2),  (0, 2, 1),  (1, 1, 1),  (1, 2, 0),  (2, 0, 1),  (2, 1, 0) \} \\           
 S_{13} &= \{ (0, 1, 2),  (0, 2, 1),  (1, 0, 2),  (1, 1, 1),  (2, 0, 1),  (2, 1, 0) \} \\            
S_{14} &= \{  (0, 1, 2),  (0, 2, 1),  (1, 0, 2),  (1, 1, 1),  (1, 2, 0),  (2, 1, 0) \} \\            
S_{15} &= \{  (0, 1, 2),  (1, 0, 2),  (1, 1, 1),  (1, 2, 0),  (2, 0, 1),  (2, 1, 0) \} \\            
S_{16} &= \{  (0, 2, 1),  (1, 0, 2),  (1, 1, 1), (1, 2, 0),  (2, 0, 1),  (2, 1, 0) \} \\             
S_{17} &= \{  (0, 1, 2),  (0, 2, 1),  (1, 0, 2),  (1, 1, 1),  (1, 2, 0),  (2, 0, 1) \} \\            
S_{18} &= \{ (0, 1, 2),  (0, 2, 1),  (1, 0, 2),  (1, 1, 1),  (1, 2, 0),  (2, 0, 1),  (2, 1, 0) \} .
\end{align*}

\begin{theorem}
There is a $P4$-automaton of 144 states that,
on input a $P4$ representation of $n$, computes
the abelian complexity of $n$.
\end{theorem}

\begin{proof}
To create the automaton, use the following {\tt Walnut} code:
\begin{verbatim}
def a001 "?msd_pisot4 Ei,x,y,z $pcount0(i,n,x) & $min0(n,x) &
   $pcount1(i,n,y) & $min1(n,y)& $pcount2(i,n,z+1) & $min2(n,z)":
#6 states

def a002 "?msd_pisot4 Ei,x,y,z $pcount0(i,n,x) & $min0(n,x) &
   $pcount1(i,n,y) & $min1(n,y) & $pcount2(i,n,z+2) & $min2(n,z)":
#125 states

def a010 "?msd_pisot4 Ei,x,y,z $pcount0(i,n,x) & $min0(n,x) &
   $pcount1(i,n,y+1) & $min1(n,y) & $pcount2(i,n,z) & $min2(n,z)":
#6 states

def a011 "?msd_pisot4 Ei,x,y,z $pcount0(i,n,x) & $min0(n,x) &
   $pcount1(i,n,y+1) & $min1(n,y) & $pcount2(i,n,z+1) & $min2(n,z)":
# 132 states

def a012 "?msd_pisot4 Ei,x,y,z $pcount0(i,n,x) & $min0(n,x) &
   $pcount1(i,n,y+1) & $min1(n,y) & $pcount2(i,n,z+2) & $min2(n,z)":
# 129 states

def a020 "?msd_pisot4 Ei,x,y,z $pcount0(i,n,x) & $min0(n,x) &
   $pcount1(i,n,y+2) & $min1(n,y) & $pcount2(i,n,z) & $min2(n,z)":
#126 states

def a021 "?msd_pisot4 Ei,x,y,z $pcount0(i,n,x) & $min0(n,x) &
   $pcount1(i,n,y+2) & $min1(n,y) & $pcount2(i,n,z+1) & $min2(n,z)":
#131 states

def a100 "?msd_pisot4 Ei,x,y,z $pcount0(i,n,x+1) & $min0(n,x) &
   $pcount1(i,n,y) & $min1(n,y) & $pcount2(i,n,z) & $min2(n,z)":
# 6 states

def a101 "?msd_pisot4 Ei,x,y,z $pcount0(i,n,x+1) & $min0(n,x) &
   $pcount1(i,n,y) & $min1(n,y) & $pcount2(i,n,z+1) & $min2(n,z)":
# 132 states

def a102 "?msd_pisot4 Ei,x,y,z $pcount0(i,n,x+1) & $min0(n,x) &
   $pcount1(i,n,y) & $min1(n,y) & $pcount2(i,n,z+2) & $min2(n,z)":
# 127 states

def a110 "?msd_pisot4 Ei,x,y,z $pcount0(i,n,x+1) & $min0(n,x) &
   $pcount1(i,n,y+1) & $min1(n,y) & $pcount2(i,n,z) & $min2(n,z)":
# 132 states

def a111 "?msd_pisot4 Ei,x,y,z $pcount0(i,n,x+1) & $min0(n,x) &
   $pcount1(i,n,y+1) & $min1(n,y) & $pcount2(i,n,z+1) & $min2(n,z)":
# 131 states

def a120 "?msd_pisot4 Ei,x,y,z $pcount0(i,n,x+1) & $min0(n,x) &
   $pcount1(i,n,y+2) & $min1(n,y) & $pcount2(i,n,z) & $min2(n,z)":
# 134 states

def a200 "?msd_pisot4 Ei,x,y,z $pcount0(i,n,x+2) & $min0(n,x) &
   $pcount1(i,n,y) & $min1(n,y) & $pcount2(i,n,z) & $min2(n,z)":
# 110 states

def a201 "?msd_pisot4 Ei,x,y,z $pcount0(i,n,x+2) & $min0(n,x) &
   $pcount1(i,n,y) & $min1(n,y) & $pcount2(i,n,z+1) & $min2(n,z)":
# 131 states

def a210 "?msd_pisot4 Ei,x,y,z $pcount0(i,n,x+2) & $min0(n,x) &
   $pcount1(i,n,y+1) & $min1(n,y) & $pcount2(i,n,z) & $min2(n,z)":
# 127 states

def num1 "?msd_pisot4 $a001(n) & ~$a002(n) & $a010(n) & ~$a011(n) & ~$a012(n) &
   ~$a020(n) & ~$a021(n) & $a100(n) & ~$a101(n) & ~$a102(n) & ~$a110(n) & 
   ~$a111(n) & ~$a120(n) & ~$a200(n) & ~$a201(n) & ~$a210(n)":

def num2 "?msd_pisot4 ~$a001(n) & ~$a002(n) & ~$a010(n) & $a011(n) & ~$a012(n) &
   ~$a020(n) & ~$a021(n) & ~$a100(n) & $a101(n) & ~$a102(n) & $a110(n) & 
   ~$a111(n) & ~$a120(n) & ~$a200(n) & ~$a201(n) & ~$a210(n)":

def num3 "?msd_pisot4 ~$a001(n) & ~$a002(n) & ~$a010(n) & $a011(n) & ~$a012(n) &
   ~$a020(n) & ~$a021(n) & ~$a100(n) & $a101(n) & ~$a102(n) & $a110(n) & 
   ~$a111(n) & ~$a120(n) & $a200(n) & ~$a201(n) & ~$a210(n)":

def num4 "?msd_pisot4 ~$a001(n) & $a002(n) & ~$a010(n) & $a011(n) & ~$a012(n) &
   ~$a020(n) & ~$a021(n) & ~$a100(n) & $a101(n) & ~$a102(n) & $a110(n) & 
   ~$a111(n) & ~$a120(n) & ~$a200(n) & ~$a201(n) & ~$a210(n)":

def num5 "?msd_pisot4 ~$a001(n) & ~$a002(n) & ~$a010(n) & $a011(n) & ~$a012(n) &
   $a020(n) & ~$a021(n) & ~$a100(n) & $a101(n) & ~$a102(n) & $a110(n) & 
   ~$a111(n) & ~$a120(n) & ~$a200(n) & ~$a201(n) & ~$a210(n)":

def num6 "?msd_pisot4 ~$a001(n) & $a002(n) & ~$a010(n) & $a011(n) & ~$a012(n) &
   ~$a020(n) & ~$a021(n) & ~$a100(n) & $a101(n) & ~$a102(n) & $a110(n) & 
   ~$a111(n) & ~$a120(n) & $a200(n) & ~$a201(n) & ~$a210(n)":

def num7 "?msd_pisot4 ~$a001(n) & ~$a002(n) & ~$a010(n) & ~$a011(n) & $a012(n) &
   ~$a020(n) & ~$a021(n) & ~$a100(n) & ~$a101(n) & $a102(n) & ~$a110(n) & 
   $a111(n) & ~$a120(n) & ~$a200(n) & $a201(n) & $a210(n)":

def num8 "?msd_pisot4 ~$a001(n) & ~$a002(n) & ~$a010(n) & ~$a011(n) & ~$a012(n) &
   ~$a020(n) & $a021(n) & ~$a100(n) & ~$a101(n) & ~$a102(n) & ~$a110(n) & 
   $a111(n) & $a120(n) & ~$a200(n) & $a201(n) & $a210(n)":

def num9 "?msd_pisot4 ~$a001(n) & ~$a002(n) & ~$a010(n) & $a011(n) & ~$a012(n) &
   $a020(n) & ~$a021(n) & ~$a100(n) & $a101(n) & ~$a102(n) & $a110(n) & 
   ~$a111(n) & ~$a120(n) & $a200(n) & ~$a201(n) & ~$a210(n)":

def num10 "?msd_pisot4 ~$a001(n) & $a002(n) & ~$a010(n) & $a011(n) & ~$a012(n) &
   $a020(n) & ~$a021(n) & ~$a100(n) & $a101(n) & ~$a102(n) & $a110(n) & 
   ~$a111(n) & ~$a120(n) & ~$a200(n) & ~$a201(n) & ~$a210(n)":

def num11 "?msd_pisot4 ~$a001(n) & ~$a002(n) & ~$a010(n) & ~$a011(n) & $a012(n) &
   ~$a020(n) & $a021(n) & ~$a100(n) & ~$a101(n) & $a102(n) & ~$a110(n) & 
   $a111(n) & $a120(n) & ~$a200(n) & ~$a201(n) & ~$a210(n)":

def num12 "?msd_pisot4 ~$a001(n) & ~$a002(n) & ~$a010(n) & ~$a011(n) & $a012(n) &
   ~$a020(n) & $a021(n) & ~$a100(n) & ~$a101(n) & ~$a102(n) & ~$a110(n) & 
   $a111(n) & $a120(n) & ~$a200(n) & $a201(n) & $a210(n)":

def num13 "?msd_pisot4 ~$a001(n) & ~$a002(n) & ~$a010(n) & ~$a011(n) & $a012(n) &
   ~$a020(n) & $a021(n) & ~$a100(n) & ~$a101(n) & $a102(n) & ~$a110(n) & 
   $a111(n) & ~$a120(n) & ~$a200(n) & $a201(n) & $a210(n)":

def num14 "?msd_pisot4 ~$a001(n) & ~$a002(n) & ~$a010(n) & ~$a011(n) & $a012(n) &
   ~$a020(n) & $a021(n) & ~$a100(n) & ~$a101(n) & $a102(n) & ~$a110(n) & 
   $a111(n) & $a120(n) & ~$a200(n) & ~$a201(n) & $a210(n)":

def num15 "?msd_pisot4 ~$a001(n) & ~$a002(n) & ~$a010(n) & ~$a011(n) & $a012(n) &
   ~$a020(n) & ~$a021(n) & ~$a100(n) & ~$a101(n) & $a102(n) & ~$a110(n) & 
   $a111(n) & $a120(n) & ~$a200(n) & $a201(n) & $a210(n)":

def num16 "?msd_pisot4 ~$a001(n) & ~$a002(n) & ~$a010(n) & ~$a011(n) & ~$a012(n) &
   ~$a020(n) & $a021(n) & ~$a100(n) & ~$a101(n) & $a102(n) & ~$a110(n) & 
   $a111(n) & $a120(n) & ~$a200(n) & $a201(n) & $a210(n)":

def num17 "?msd_pisot4 ~$a001(n) & ~$a002(n) & ~$a010(n) & ~$a011(n) & $a012(n) &
   ~$a020(n) & $a021(n) & ~$a100(n) & ~$a101(n) & $a102(n) & ~$a110(n) & 
   $a111(n) & $a120(n) & ~$a200(n) & $a201(n) & ~$a210(n)":

def num18 "?msd_pisot4 ~$a001(n) & ~$a002(n) & ~$a010(n) & ~$a011(n) & $a012(n) &
   ~$a020(n) & $a021(n) & ~$a100(n) & ~$a101(n) & $a102(n) & ~$a110(n) & 
   $a111(n) & $a120(n) & ~$a200(n) & $a201(n) & $a210(n)":

eval coverall "?msd_pisot4 An (n>=1) => ($num1(n)|$num2(n)|$num3(n)|$num4(n)|$num5(n)|
   $num6(n)|$num7(n)|$num8(n)|$num9(n)|$num10(n)|$num11(n)|$num12(n)|
   $num13(n)|$num14(n)|$num15(n)|$num16(n)|$num17(n)|$num18(n))":

combine pab num1 num2 num3 num4 num5 num6 num7 num8 num9 num10 num11 num12 
   num13 num14 num15 num16 num17 num18:

morphism abc "0->1 1->3 2->3 3->4 4->4 5->4 6->5 7->5 8->5 9->5 10->5 11->5 
   12->6 13->6 14->6 15->6 16->6 17->6 18->7":

image BC3 abc pab:
\end{verbatim}
\end{proof}

\begin{corollary}
The abelian complexity of $\bf p$, for $n \geq 1$,
lies in $\{3,4,5,6,7 \}$.  Each possibility occurs
infinitely often.
\end{corollary}

\section{Palindromes}

\begin{theorem}
The only palindromes occurring in $\bf p$ are
$$\{ 0,  1,  2,  121,  101,  010,  01210,  21012,  1012101 \}.$$
\end{theorem}

\begin{proof}
It suffices to list all the factors of
length $\leq 9$, since any longer palindrome would have these as factors.  We can be sure we have examined all of them, by Theorem~\ref{subword}.
\end{proof}

\section{Bispecial factors}

We say a factor $w$ of an infinite word $\bf x$
is {\it right-special} (resp., {\it left-special}) if there exist two distinct
letters $a, b$ such that $wa$ and $wb$ (resp., $aw$ and $bw$) are both factors of $\bf x$.   A word $w$ is {\it bispecial} if it is both right- and left-special.

We can determine the lengths of bispecial factors
occurring in $\bf p$.
\begin{verbatim}
def pisotrightspec "?msd_pisot4 Ej $pisi(i,j,n) & PI[i+n]!=PI[j+n]":
def pisotleftspec "?msd_pisot4 Ej $pisi(i,j,n) & PI[i-1]!=PI[j-1]":
def pisotbispec "?msd_pisot4 $pisotrightspec(i,n) & $pisotleftspec(i,n)":
def pisotbispeclen "?msd_pisot4 Ei $pisotbispec(i,n)":
\end{verbatim}
The result is displayed in Figure~\ref{bispec}.
\begin{figure}[H]
\begin{center}
\includegraphics[width=6in]{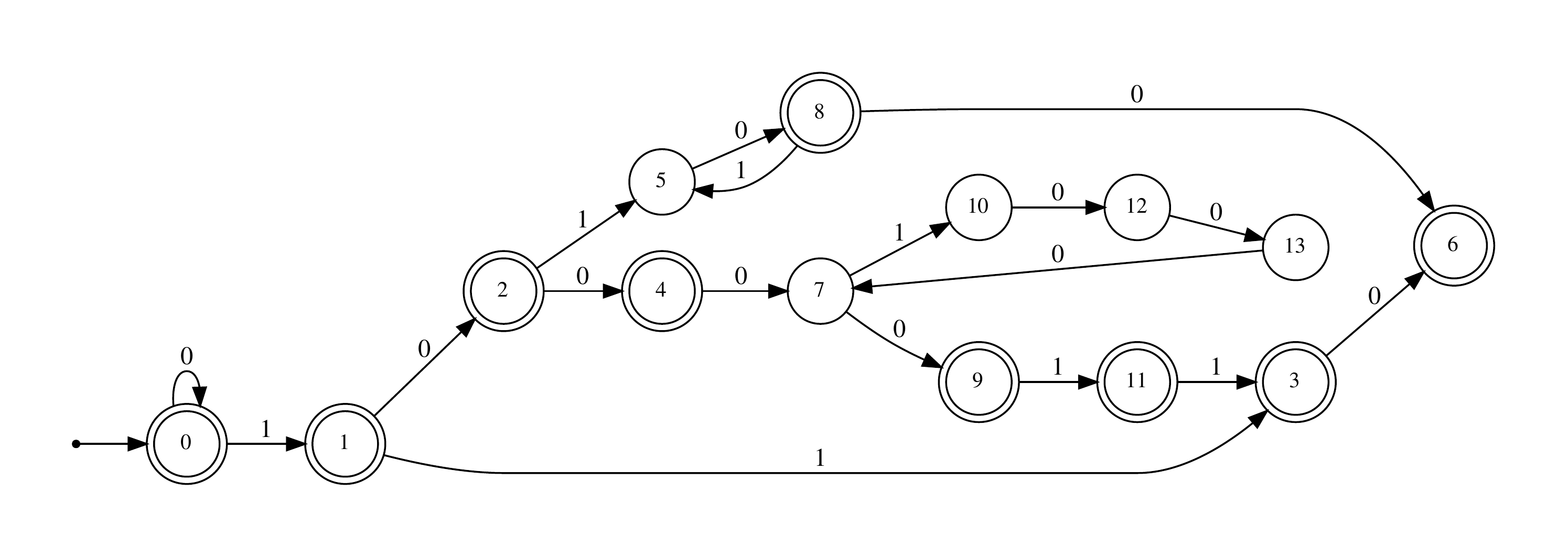}
\end{center}
\caption{Automaton recognizing lengths of bispecial factors in $\bf p$.}
\label{bispec}
\end{figure}
\begin{theorem}
The word $\bf p$ has a bispecial factor of length $n$ iff $$(n)_P \in \{1, 11, 110 \} \, \cup\, (10)^+ \{\epsilon, 0 \} \, \cup \, (1000)^+ \{ 0, 01, 011, 0110 \} .$$
\end{theorem}

\section{Letter density}

In this section we obtain the letter densities
of $\bf p$.  We know from 
\cite{Cassaigne&Labbe&Leroy:2017} that the densities exist, and therefore it suffices to determine
them on prefixes of the form $h^n(0)$.  

\begin{theorem}
\leavevmode
\begin{itemize}
    \item The density of $0$ is $1/\beta_1^2 \doteq 0.324717957244746$;
    \item The density of $1$ is $1/\beta_1^2 + 1/\beta_1^4 \doteq 0.430159709001946734$;
    \item The density of $2$ is $1/\beta_1^3 + 1/\beta_1^5 \doteq 0.2451223337533$.
\end{itemize}
\end{theorem}

\begin{proof}
An easy induction gives
$\psi(|h^n(0)|) = (X_{n-1}, X_{n-1}+X_{n-3}, X_{n-2}+X_{n-4})$,
from which the desired result follows immediately.
\end{proof}

\begin{remark}
These densities were also given in
\cite{Cassaigne&Labbe&Leroy:2022} for a slightly 
different morphism with the same incidence matrix.
\end{remark}

\section{Recurrence and appearance}

Let $\bf x$ be an infinite word.  The recurrence
function $R(n)$ is defined to be the smallest
positive integer $m$ such that every occurrence
of a length-$n$ factor $x$ is followed by another occurrence of the same word at distance at most
$R(n)$.  We now compute it for the word $\bf p$.

\begin{theorem}
Define the sequences $(B_i)$ and $(C_i)$ as follows:
\begin{align*}
B_{2i+3} &= [1(01)^i 10]_P  \\
B_{2i+4} &= [1(01)^i 100]_P \\
C_{4i+2} &= [10(0010)^i]_P \\
C_{4i+3} &= [10(0010)^i 0]_P \\
C_{4i+4} &= [10(0010)^i 00]_P \\
C_{4i+5} &= [10(0010)^i 001]_P
\end{align*}
for $i \geq 0$.
Then 
$$R(n) = \begin{cases}
5, & \text{if $n = 1$;} \\
12, & \text{if $n=2$; } \\
16, & \text{if $n=3$; } \\
21, & \text{if $n=4$; } \\
28, & \text{if $n=5$; } \\
X_{i+4}, & \text{if 
$B_i \leq n \leq C_{i+1}$ for $i \geq 3$;}\\
X_{i+4}+X_{i+2}, & \text{if 
$C_{i+1} < n < B_{i+1}$ for $i \geq 3$.}
\end{cases}
$$
\label{recur}
\end{theorem}

\begin{proof}
We use the following {\tt Walnut} code to compute
$R$. 
\begin{verbatim}
def poccur "?msd_pisot4 Ai Ej j>i & j<=i+x & $pisi(i,j,n)":
def precur "?msd_pisot4 $poccur(n,x) & ~$poccur(n,x-1)":
reg pix msd_pisot4 "0*10*": 
def precB "?msd_pisot4 Ex $pix(x) & $precur(n,x) & ~$precur(n-1,x)":
def precC "?msd_pisot4 Ex $pix(x) & $precur(n,x) & ~$precur(n+1,x)":
\end{verbatim}
\end{proof}

\begin{corollary}
We have $R(n) \leq \xi n$ for
$\xi = (2\beta_1^2-\beta_1 + 2) \doteq 6.40431358$.
\end{corollary}

The appearance
function $A(n)$ is defined to be the smallest
positive integer $m$ such that every occurrence
of a length-$n$ factor $x$ begins at a position
$\leq m$.
We now compute it for the word $\bf p$.
\begin{theorem}
Define the sequences $(B_n)$ and $(C_n)$ as
in Theorem~\ref{recur}.
We have
$$A(n) = 
\begin{cases}
2, & \text{if $n= 1$;} \\
4, & \text{if $n=2$; } \\
7, & \text{if $n=3$; } \\
11, & \text{if $n = 4$;} \\
13, & \text{if $n = 5$;} \\
X_{2i+4} - 1, & \text{if $B_{2i+1} < n \leq
    C_{2i+2}$ for $i \geq 1$;} \\
C_{2i+4}, & \text{if $C_{2i+2} < n < B_{2i+2}$
    for $i \geq 1$;} \\
X_{2i+5} - 1, & \text{if $B_{2i+2} \leq n \leq  
C_{2i+3}$ for $i \geq 1$}; \\
C_{2i+5}, & \text{if $C_{2i+3} < n \leq B_{2i+3}$ for $i \geq 1$}.
\end{cases}
$$
\end{theorem}

\begin{corollary}
We have $A(n) \leq \zeta n$ for
$\zeta = (2\beta_1^2-2\beta_1 + 1) \doteq 3.6494359$.
\end{corollary}

\section{Characterization of factors}
\label{characterization}

Using the convention in~\cite{Badkobeh&Ochem:2015}, we say that a factorial
language $L$ over $\Sigma_k$ defined by forbidden factors \emph{characterizes}
a morphic word $w$ if $w$ belongs to $L$ and every recurrent
factor of an infinite word in $L$ is a factor of $w$.

\begin{theorem}\label{lemmap}
The language $L$ of cube-free ternary words avoiding\\
$F = \{00, 11, 22, 20, 212, 0101, 02102, 121012, 01021010, 21021012102\}$ characterizes the word~${\bf p}$.
\end{theorem}

\begin{proof}
First, we check that $L$ contains $\bf p$. Indeed, ${\bf p}$ contains no
factor in $F$ and Theorem~\ref{gamma} implies that ${\bf p}$ is cube-free.

Now consider a bi-infinite word $w$ in $L$.
Since $w$ avoids $\{00, 11, 22, 20, 212\}$, the only factors of $w$ of the form $0z0$
with $z\in\{1,2\}^*$ are $010$, $01210$, $0210$.
So $w\in\{01,0121,021\}^\omega$ and thus $w\in\{01,21,0\}^\omega$.
So we write $w=h(v)$ where $h$ is the morphism
$0 \rightarrow 01$, $1 \rightarrow 21$, $2 \rightarrow 0$.

Now it suffices to show that $L$ contains $v$ too.
Since $w$ is cube-free, its pre-image $v$ is also cube-free.

To show that $v$ avoids $F$, we consider every $f\in F$
and we show by contradiction that $f$ is not a factor $v$.

\begin{enumerate}[(a)]
 \item if $v$ contains $00$, then $h(00)=0101\in F$.\label{00}
 \item if $v$ contains $11$, then $h(11)=2121$ contains $212\in F$.\label{11}
 \item if $v$ contains $22$, then $h(22)=00\in F$.\label{22}
 \item if $v$ contains $20$, then $h(20)=001$ contains $00\in F$.\label{20}
 \item if $v$ contains $212$, then $v$ contains $2121$ by (\ref{22}) and (\ref{20}).\\ $h(2121)=021021$ contains $02102\in F$.\label{212}
 \item if $v$ contains $0101$, then $h(0101)=01210121$ contains $121012\in F$.\label{0101}
 \item if $v$ contains $02102$, then $h(02102)=01021010\in F$.\label{02102}
 \item if $v$ contains $121012$, then $v$ contains $1210121$ by (\ref{22}) and (\ref{20}).\\ $h(1210121)=210210121021$ contains $21021012102\in F$.\label{121012}
 \item if $v$ contains $01021010$, then $v$ contains $010210102$ by (\ref{00}) and (\ref{0101}).\\
 $v$ contains $0102101021$ by (\ref{22}) and (\ref{20}).\\
 $v$ contains $01021010210$ by (\ref{11}) and (\ref{212}).\\
 $v$ contains $010210102101$ by (\ref{00}) and (\ref{02102}).\\
 $v$ contains $1010210102101$ by (\ref{00}) and (\ref{20}).\\
 $v$ contains $21010210102101$ by (\ref{0101}) and (\ref{11}).\\
 $v$ contains $210102101021012$ by (\ref{11}) and to avoid $(21010)^3$.\\
 $v$ contains $1210102101021012$ by (\ref{22}) and to avoid $(02101)^3$.\\
 $h(1210102101021012)=2102101210102101210102101210=2(102101210)^3$.
 \item if $v$ contains $21021012102$, then $v$ contains $121021012102$ by (\ref{02102}) and (\ref{22}).\\
 $v$ contains $0121021012102$ by (\ref{00}) and (\ref{212}).\\
 $v$ contains $10121021012102$ by (\ref{00}) and (\ref{20}).\\
 $v$ contains $210121021012102$ by (\ref{0101}) and (\ref{11}).\\
 $v$ contains $0210121021012102$ by (\ref{121012}) and (\ref{22}).\\
 $v$ contains $10210121021012102$ by (\ref{00}) and (\ref{20}).\\
 $v$ contains $102101210210121021$ by (\ref{22}) and (\ref{20}).\\
 $v$ contains $1021012102101210210$ by (\ref{11}) and (\ref{212}).\\
 $v$ contains $10210121021012102101$ by (\ref{00}) and (\ref{02102}).\\
 $v$ contains $102101210210121021010$ by (\ref{11}) and to avoid $(1021012)^3$.\\
 $h(102101210210121021010)=2101021012102101021012102101021012101=(210102101210)^31$.
\end{enumerate}
\end{proof}





\begin{thebibliography}{10}

\bibitem{Badkobeh&Ochem:2015}
G.~Badkobeh and P.~Ochem.
\newblock Characterization of some binary words with few squares.
\newblock {\em Theoret. Comput. Sci.} {\bf 588} (2015), 73--80.

\bibitem{Berstel:1995}
J.~Berstel.
\newblock {\em Axel {Thue's} Papers on Repetitions in Words: a Translation}.
\newblock Number~20 in Publications du Laboratoire de Combinatoire et
  d'Informatique {Math\'ematique}. Universit\'e du Qu\'ebec \`a Montr\'eal,
  February 1995.

\bibitem{Bruyere&Hansel:1997}
V.~{Bruy\`ere} and G.~Hansel.
\newblock Bertrand numeration systems and recognizability.
\newblock {\em Theoret. Comput. Sci.} {\bf 181} (1997), 17--43.

\bibitem{Cassaigne&Labbe&Leroy:2017}
J.~Cassaigne, S.~{Labb\'e}, and J.~Leroy.
\newblock A set of sequences of complexity $2n + 1$.
\newblock In S.~Brlek et~al., editors, {\em WORDS 2017}, Vol. 10432 of {\em
  Lecture Notes in Computer Science}, pp.  144--156. Springer-Verlag, 2017.

\bibitem{Cassaigne&Labbe&Leroy:2022}
J.~Cassaigne, S.~{Labb\'e}, and J.~Leroy.
\newblock Almost everywhere balanced sequences of complexity $2n+1$.
\newblock ArXiv preprint arXiv:2102.10093 [math.DS], available at
  \url{https://arxiv.org/abs/2102.10093}, 2022.

\bibitem{Coven&Hedlund:1973}
E.~M. Coven and G.~A. Hedlund.
\newblock Sequences with minimal block growth.
\newblock {\em Math. Systems Theory} {\bf 7} (1973), 138--153.

\bibitem{Currie&Mol&Ochem&Rampersad&Shallit:2022}
J.~D. Currie, L.~Mol, P.~Ochem, N.~Rampersad, and J.~Shallit.
\newblock Complement avoidance in binary words.
\newblock In preparation, 2022.

\bibitem{Everest&vanderPoorten&Shparlinski&Ward:2003}
G.~Everest, A.~{van der Poorten}, I.~Shparlinski, and T.~Ward.
\newblock {\em Recurrence Sequences}, Vol. 104 of {\em Mathematical Surveys and
  Monographs}.
\newblock Amer. Math. Soc., 2003.

\bibitem{Fraenkel:1985}
A.~S. Fraenkel.
\newblock Systems of numeration.
\newblock {\em Amer. Math. Monthly} {\bf 92} (1985), 105--114.

\bibitem{Frougny&Solomyak:1996}
C.~Frougny and B.~Solomyak.
\newblock On representation of integers in linear numeration systems.
\newblock In M.~Pollicott and K.~Schmidt, editors, {\em Ergodic Theory of
  {$\Zee^d$} Actions (Warwick, 1993--1994)}, Vol. 228 of {\em London
  Mathematical Society Lecture Note Series}, pp.  345--368. Cambridge
  University Press, 1996.

\bibitem{Glen&Justin:2009}
A.~Glen and J.~Justin.
\newblock Episturmian words: a survey.
\newblock {\em RAIRO Inform. Th\'eor. App.} {\bf 43} (2009), 403--442.

\bibitem{Karhumaki:1983}
J.~{Karhum\"aki}.
\newblock On cube-free $\omega$-words generated by binary morphisms.
\newblock {\em Disc. Appl. Math.} {\bf 5} (1983), 279--297.

\bibitem{Mignosi&Pirillo:1992}
F.~Mignosi and G.~Pirillo.
\newblock Repetitions in the {Fibonacci} infinite word.
\newblock {\em RAIRO Inform. Th\'eor. App.} {\bf 26} (1992), 199--204.

\bibitem{Morse&Hedlund:1940}
M.~Morse and G.~A. Hedlund.
\newblock Symbolic dynamics {II}. {Sturmian} trajectories.
\newblock {\em Amer. J. Math.} {\bf 62} (1940), 1--42.

\bibitem{Mousavi:2016}
H.~Mousavi.
\newblock Automatic theorem proving in {{\tt Walnut}}.
\newblock Arxiv preprint arXiv:1603.06017 [cs.FL], available at
  \url{http://arxiv.org/abs/1603.06017}, 2016.

\bibitem{Mousavi&Schaeffer&Shallit:2016}
H.~Mousavi, L.~Schaeffer, and J.~Shallit.
\newblock Decision algorithms for {Fibonacci}-automatic words, {I}: {Basic}
  results.
\newblock {\em RAIRO Inform. Th\'eor. App.} {\bf 50} (2016), 39--66.

\bibitem{Mousavi&Shallit:2015}
H.~Mousavi and J.~Shallit.
\newblock Mechanical proofs of properties of the {Tribonacci} word.
\newblock In F.~Manea and D.~Nowotka, editors, {\em Proc. WORDS 2015}, Vol.
  9304 of {\em Lecture Notes in Computer Science}, pp.  1--21. Springer-Verlag,
  2015.

\bibitem{Richomme&Saari&Zamboni:2011}
G.~Richomme, K.~Saari, and L.~Q. Zamboni.
\newblock Abelian complexity in minimal subshifts.
\newblock {\em J. London Math. Soc.} {\bf 83} (2011), 79--95.

\bibitem{Shallit:2021f}
J.~Shallit.
\newblock Abelian complexity and synchronization.
\newblock {\em INTEGERS --- Elect. J. Comb. Numb. Theory} {\bf 21} (2021),
  \#A36 (electronic).
\newblock Available at \\ \url{http://math.colgate.edu/~integers/v36/v36.pdf}.

\bibitem{Shallit:2021h}
J.~Shallit.
\newblock Synchronized sequences.
\newblock In T.~Lecroq and S.~Puzynina, editors, {\em WORDS 2021}, Vol. 12847
  of {\em Lecture Notes in Computer Science}, pp.  1--19. Springer-Verlag,
  2021.

\bibitem{Shallit:2022}
J.~Shallit.
\newblock {\em The Logical Approach To Automatic Sequences: Exploring
  Combinatorics on Words with {\tt Walnut}}.
\newblock Cambridge University Press, 2022.
\newblock In press.

\bibitem{Thue:1912}
A.~Thue.
\newblock {\"Uber} die gegenseitige {Lage} gleicher {Teile} gewisser
  {Zeichenreihen}.
\newblock {\em Norske vid. Selsk. Skr. Mat. Nat. Kl.} {\bf 1} (1912), 1--67.
\newblock Reprinted in {\it Selected Mathematical Papers of Axel Thue}, T.
  Nagell, editor, Universitetsforlaget, Oslo, 1977, pp.~413--478.

\end{thebibliography}
\end{document}